\documentclass{article}
\usepackage{arxiv}
\usepackage{xparse}
\usepackage[pdftex]{graphicx}
\usepackage{bm}
\usepackage{amsmath,amsthm,amsfonts}
\usepackage{algorithmic}
\usepackage{tabularx}
\usepackage[caption=false,font=normalsize,labelfont=sf,textfont=sf]{subfig}

\NewDocumentCommand\norm{mg}
{\|{#1}\|\IfNoValueTF{#2}{}{_{#2}}}

\DeclareMathOperator{\tr}{tr}
\DeclareMathOperator{\qr}{qr}

\newcommand\bSigma{\bm{\Sigma}}
\newcommand\bx{\mathbf{x}}
\newcommand\by{\mathbf{y}}
\newcommand\bv{\mathbf{v}}
\newcommand\bK{\mathbf{K}}
\newcommand\bJ{\mathbf{J}}
\newcommand\bM{\mathbf{M}}
\newtheorem{clm}{Claim}
\newtheorem{lem}{Lemma}
\begin{document}
\title{The Level Set Kalman Filter for State Estimation of Continuous-discrete Systems}
\author{Ningyuan~Wang
        and~Daniel~B~Forger
\thanks{N. Wang is with the Department
of Mathematics, University of Michigan, 
Ann Arbor, MI, 48109}
\thanks{D. B. Forger is with Department of Mathematics, Department of Computational Medicine and Bioinformatics and Michigan Institute for Data Science, University of Michigan,
Ann Arbor, MI, 48109
 e-mail: forger@umich.edu
}
}

\markboth{ }%
{Wang \MakeLowercase{\textit{et al.}}: LSKF for Continuous-discrete Systems}

\maketitle

\begin{abstract}
We propose a new extension of Kalman filtering for continuous-discrete systems with nonlinear state-space models that we name as the level set Kalman filter (LSKF). The LSKF assumes the probability distribution can be approximated as a Gaussian and updates the Gaussian distribution through a time-update step and a measurement-update step. The LSKF improves the time-update step compared to existing methods, such as the continuous-discrete cubature Kalman filter (CD-CKF), by reformulating the underlying Fokker-Planck equation as an ordinary differential equation for the Gaussian, thereby avoiding the need for the explicit expression of the higher derivatives. Together with a carefully picked measurement-update method, numerical experiments show that the LSKF has a consistent performance improvement over the CD-CKF for a range of parameters. Meanwhile, the LSKF simplifies implementation, as no user-defined timestep subdivisions between measurements are required, and the spatial derivatives of the drift function are not explicitly needed.
\end{abstract}


\section*{Copyright info}
IEEE Trans. Signal Process. \emph{Early Access} (2021) 10.1109/TSP.2021.3133698

\copyright 2021 IEEE. Personal use of this material is permitted.  Permission from IEEE must be obtained for all other uses, in any current or future media, including reprinting/republishing this material for advertising or promotional purposes, creating new collective works, for resale or redistribution to servers or lists, or reuse of any copyrighted component of this work in other works.
\section{Introduction}
Kalman Filtering methods are used in many applications. A Bayesian filtering method updates a \emph{state estimation} of the target given knowledge of the system and measurements~\cite{sarkka2013bayesian}.  The goal of these methods is to estimate the state of a target system where the dynamics are known, using measurements taken a fixed time intervals, and accounting for noise or uncertainty in the system and measurements. There are two parts to these methods. First, a new measurement is used to generate the best possible estimate of the system state. Second, that estimate is propagated forward using the system's dynamics until the subsequent measurement is available. Here, we present a method for accurately implementing that second step in the presence of noise.

The general framework for these problems was first described by Kalman~\cite{kalman1961new}. A Kalman-Bucy type filtering consists of two steps: a \emph{measurement-update} part that updates the estimation using the measurement and a state estimation from previous steps, and a \emph{time-update} step that updates the state estimation between consecutive measurements. The level set Kalman filter (LSKF) method focuses on the improvement of the \emph{time-update} step, and the discussions that follow are restricted to the \emph{time-update} part unless we explicitly mention the \emph{measurement-update}.

Assuming that the dynamics are linear (in space), and all noise is Gaussian, Kalman-Bucy filtering~\cite{kalman1961new} gives an optimal way to estimate the system state for the \emph{time-update}. However, in many cases, we would like to generalize this method to a system where the dynamics of interest are nonlinear. For such a nonlinear system, a Gaussian probability density function (PDF) is no longer preserved, even when the dynamics are quadratic~\cite{gustafsson2011some}. When the dynamics are approximately linear for the region of state space where most of the PDF lies, Unscented Kalman filtering (UKF)  \cite{julier2004unscented} can provide a useful method. Additionally, the UKF is easy to implement since it does not require explicit evaluation of the Jacobian of the velocity field, which is not readily available in practical problems where, for example, the velocity field is implicitly defined. 

Researchers have improved how the process noise is incorporated, but so far, methods are significantly more complicated than the UKF (e.g., requiring the explicit calculation of a Jacobian) or only work with specific numerical solvers.
Good examples include the continuous-discrete Kalman filter~\cite{sarkka2007unscented} (CDKF) and the continuous-discrete Cubature Kalman filter~\cite{arasaratnam2010cubature} (CD-CKF). (Note the latter uses the Cubature Kalman transformation as introduced in~\cite{arasaratnam2009cubature} instead of the unscented Kalman transformation, however, it can be reformulated to use either, as explained in~\cite{kulikov2017accurate}.) The CDKF in~\cite{sarkka2007unscented} addresses the continuous nature of the process noise; however, the derivation of their method involves approximations such that their method is not exact even if the dynamics are linear. Moreover, the computation of the prediction is significantly more complicated than the original UKF, eroding its advantage the CDKF offers by removing intermediate timesteps. The CD-CKF uses a $1.5$-order It\^{o}-Taylor expansion of the stochastic differential equation, which uses the Jacobian (or approximations of it) that can be difficult to calculate. Though the explicit Jacobian can be avoided by deriving specific Runge-Kutta methods as described in~\cite{newton1991asymptotically}, this still complicates programming and limits the type of numerical solvers available. 

Here, we propose the LSKF that addresses these issues. Our method: 1) does not require the Jacobian or any spatial partial derivative of the drift function explicitly, 2) allows the use of adaptive ordinary differential equation~(ODE) solvers and frees the user from choosing the time discretization, and 3) shows performance improvements over the CD-CKF, even in the challenging test cases presented in~\cite{arasaratnam2010cubature}. From a theory point of view, our derivation of the method is based on the apparent velocity of the level set of the probability distribution, which is a novel approach to analyze these problems, and may enable further developments.


\section{Problem Statement and Background}
(Note on notation: we distinguish matrix or vector-valued quantities $\bv$ versus scalar-valued quantities $v_1$ by using a bold font. A list of symbols is included in Table \ref{lst:symbols} in the appendix. )
\subsection{Problem formulation}
A continuous dynamic discrete measurement system includes a continuous-time process described by a Fokker-Planck equation and a discrete measurement process with measurement noise.

The discrete measurement process is defined by a transformation $h$ from state space to the observation space, together with a zero-mean Gaussian observation noise $\mathbf{\tau} \sim \mathcal{N}(\mathbf{0},\mathbf{R})$. Suppose at the time of measurement, the state vector is $\bx$, then the measurement $\by$ is given by:
\begin{equation}
\by = h(\bx) + \mathbf{\tau},
\end{equation}where $\mathbf{\tau} \sim \mathcal{N}(\mathbf{0},\mathbf{R})$.

In between the time where two consecutive measurements are taken, we assume that the process noise is Gaussian, and the system equations are described by the It\^{o} process~\cite{jazwinski2007stochastic}: 
\begin{equation}
\label{eqn:ito_proc}
\frac{d \bx}{dt} = \bv(\bx) + \sqrt{\bK} \frac{d \beta}{dt},
\end{equation}
where $\bv$ is the \textbf{drift function}, or velocity field defined by the dynamics, $\beta$ is a standard $d$ dimension Brownian process. $\bK$ is an $d \times d$ positive semi-definite continuous \textbf{process noise matrix}, and $\bK = \sqrt{\bK} \sqrt{\bK}^T$.

Then, the PDF $u$ is described by the Fokker-Planck equation of the following form:

\begin{equation}
\label{eqn:heatadv_prestate}
\frac{d u}{dt} =  \frac{1}{2}\nabla \cdot \bK \nabla u - \nabla \cdot (\bv u).
\end{equation}
\subsection{Brief review of existing time-update methods}

Under the assumption that the drift function $\bv$ is linear in space and the process noise matrix $\bK$ is constant, it can be shown that a Gaussian PDF $u$ is preserved. (A proof of this fact using level sets is in the next section). In~\cite{kalman1961new}, the derivation of the \emph{time-update} step is based on this observation.

One often wants to generalize this method to nonlinear models even when Gaussian distributions are no longer exactly preserved. One generalization would be to use the Jacobian of the drift function at the mean of the distribution, which is a key part of the Extended Kalman-Bucy Filter (EKF) method. One disadvantage of the EKF is the need for an explicit formula of the Jacobian of the drift function. The UKF is also derived based on the assumption of a local linearization of velocity; however, the explicit evaluation of the Jacobian is avoided.


In~\cite{arasaratnam2010cubature}, after their comparison between the continuous-discrete cubature Kalman filter (CD-CKF), continuous-discrete unscented Kalman filter (CD-UKF), and continuous-discrete extended Kalman filter (CD-EKF), they concluded that ''the CD-CKF is the choice for challenging radar problems''. [2, p.4987] In~\cite{kulikov2017accurate}, Kulikov and Kulikova presented a new filtering method named the accurate continuous-discrete extended Kalman filter (ACD-EKF), and compared it to the CD-CKF and CD-UKF. Note the implementation of the CD-UKF in~\cite{kulikov2017accurate} is more sophisticated than that in~\cite{arasaratnam2010cubature} as it uses the IT-1.5 that is the same as presented in~\cite{arasaratnam2010cubature} for the CD-CKF. With the improved implementation of the CD-UKF, Kulikov and Kulikova reported in~\cite{kulikov2017accurate} that the CD-CKF and the CD-UKF perform similarly. In addition, while the ACD-EKF requires less tuning than the CD-CKF, with sufficient timestep subdivision, the CD-CKF seems to outperform the ACD-EKF, as stated in the conclusion of~\cite{kulikov2017accurate}: \emph{The highest accuracy is provided by the most time-consuming filters CD-CKF256 and CD-UKF256.} Therefore, we conclude that with sufficient timestep subdivision, the CD-CKF is still a benchmark method to compare against.

\subsection{the time-update of the CD-CKF with Ito-Taylor expansion}
\label{ssec:CD-CKF}
In~\cite{arasaratnam2010cubature}, the Ito-Taylor expansion of order $1.5$~(IT-1.5) is first introduced to the \emph{time-update} step of the continuous-discrete filtering. It is confirmed in~\cite{sarkka2012continuous} that the Unscented Kalman filtering with IT-1.5 achieves similar performance to the CD-CKF. For the purpose of comparing time-update, the performance of the CD-CKF should suffice for a benchmark. Additionally, we noted that while the IT-1.5 should converge to the accurate result with a weak order of convergence $2$, the implementation in~\cite{arasaratnam2010cubature} chooses to only discretize noise once between the measurements and does not converge to this result, presumably as a tradeoff to improve speed. For the sake of complete comparison, we also implemented a version with a proper IT-1.5 expansion that discretizes noise for every timestep subdivision.

Here we restate the \emph{square-root form} of the CD-CKF, as derived in~\cite{arasaratnam2010cubature}.

\emph{Time-update}: For update with a timestep of $\Delta t$, we define the function
\begin{equation}
\mathbf{f}_d(\bx,t) := \bx(t) + \Delta t \bv(\bx(t),t)+\frac{1}{2}\Delta t^2(\mathbb{L}_0(\bv(\bx,t))),
\end{equation}
where the opertor $\mathbb{L}_0$ is defined as
\begin{align}
\mathbb{L}_0 &:= \frac{d}{d t} + \sum_{i=1}^d v_i\frac{\partial}{\partial x_i}\\
&+\frac{1}{2}\sum_{j,p,q=1}^d\sqrt{\bK}_{p,j}\sqrt{\bK}_{j,q}\frac{\partial^2}{\partial x_p \partial x_q}\nonumber.
\end{align}
We also define the operator $\mathbb{L}(\bv)$ be the square matrix defined entrywisely with its $(i,j)$th element being $\mathbb{L}_j v_i$, where 
\begin{equation}
\mathbb{L}_j := \sum_{i=1}^d \sqrt{\bK}_{i,j}\frac{\partial}{\partial x_i}.
\end{equation}
Then the time-update algorithm is as follows: 
\begin{algorithmic}[1]
\REQUIRE{Guess of initial state $\bx_0$ at time $t_0$, and a factorization of a guess of covariance matrix $\bM$.
Drift velocity $\bv$, continuous process noise matrix $\bK$.}
\STATE{Find the $2d$ cubature points $\bx_i = \bx_0 + (\bM)_i, \bx_{i+d} = \bx_0 - (\bM)_i$, where $i=1,\dots,d$, and $(\bM)_i$ denotes the $i$th column of $\bM$.}
\STATE{Evaluate the propagated cubature point: $\bx^*_i = \mathbf{f}_d(\bx_i,t_0)$, where where $i=1,\dots,2d$.}
\STATE{Estimate the updated mean by the average of the cubature points: 
\begin{equation}
\bx^*_0 := \frac{1}{2d} \bx^*_i.
\end{equation}}
\STATE{Estimate the factorization of the covariance matrix by the triangularization of the concatenated matrix:
\begin{equation}
\bM^* = \text{tria}\left(
\begin{bmatrix}
\mathbf{X}^*| \sqrt{\Delta t}(\sqrt{\bK}+\frac{\Delta t}{2}\mathbb{L}(\bv))|\sqrt{\frac{\Delta t^3}{12}\mathbb{L}(\bv)}
\end{bmatrix}
\right),
\end{equation}
where $\text{tria}(\cdot)$ denotes applying a triangularization procedure such as the Gram-Schmidt based QR-decomposition, $\mathbf{X}^*$ is a matrix with $i$th column being $\bx^*_i$, $\mathbb{L}(\bv) = \mathbb{L}(\bv(\bx^*_0,t_0))$. 
}
\RETURN{Estimated updated mean $\bx^*_0$ and updated factorization of the covariance matrix $\bM^*$ at $t_0 + \Delta t$.}
\end{algorithmic}
\subsection{The square root form of cubature Kalman measurement-update}
\label{ssec:sqcdckf}
Here, we discuss the \emph{measurement-update} method used in the CD-CKF and the LSKF. Since the operations from the time-update can cause $\bM$ to be positive semi-definite, a measurement-update method that can accommodate a positive semi-definite matrix is required for reliability, as pointed out in~\cite{arasaratnam2010cubature}. We used the measurement-update method from the square root CD-CKF method, as stated in Appendix B of~\cite{arasaratnam2010cubature}.  Since the notations used are different, the measurement-update of the square root CD-CKF is restated here for reference.
\begin{algorithmic}[1]
\REQUIRE{Factorization of the predicted covariance matrix before measurement $\bM$, predicted mean before measurement $\bar{\bx}$, measurement $\mathbf{y}$, 
a factorization of the covariance of the measurement noise matrix $\sqrt{\mathbf{R}}$, measurement function $\mathbf{h}$}
\STATE{Find the concatenated cubature points matrix of size $d\times 2d$:
\begin{equation}
\mathbf{N} = \bar{\bx} + \sqrt{2d}\begin{bmatrix}
\bM | -\bM
\end{bmatrix},
\end{equation}
where the vector-matrix addition is applied as $\bar{\bx}$ added to each \textbf{column} of the concatenated matrix $[\bM | -\bM]$
}
\STATE{Evaluated the propagated cubature points
\begin{equation}
\mathbf{Y} = \mathbf{h}(\mathbf{N}),
\end{equation}
where the measurement function $\mathbf{h}(\cdot)$ is evaluated on each \textbf{column}.}
\STATE{
Estimate the predicted measurement 
\begin{equation}
\bar{\by} = \frac{1}{2d}\sum_{i=1}^{2d} \mathbf{Y}_i.
\end{equation}
}
\STATE{
compute matrices $\mathbf{T}_{11}$,$\mathbf{T}_{21}$, and $\mathbf{T}_{22}$ by the following QR-factorization:
\begin{equation}
\begin{bmatrix}
\mathbf{T}_{11}&\mathbf{O}\\ \mathbf{T}_{21} & \mathbf{T}_{22}
\end{bmatrix}= \qr
\left(
\begin{bmatrix}
\mathbf{Y}&\sqrt{\mathbf{R}}\\
\mathbf{N}&\mathbf{O}
\end{bmatrix}
\right),
\end{equation}
where $\mathbf{O}$ denotes a zero matrix of appropriate size. 
}
\STATE{
Estimate the cubature gain
\begin{equation}
\mathbf{W} = \mathbf{T}_{21} / \mathbf{T}_{11},
\end{equation}
where $/$ represents solving for $\mathbf{W}$ in $\mathbf{T}_{21} = \mathbf{W} \mathbf{T}_{11}$ using a backward stable solver.
}
\STATE{
Estimate the mean of the corrected state
\begin{equation}
\hat{\bx} = \bar{\bx} + \mathbf{W}(\by - \bar{\by}).
\end{equation}
}
\STATE{
Estimate a factorization of the corrected covariance matrix
\begin{equation}
\hat{\bM} = \mathbf{T}_{22}.
\end{equation}
}
\RETURN{Corrected mean $\hat{\bx}$ and a factorization of the corrected covariance matrix $\hat{\bM}$.}
\end{algorithmic}

\section{Derivation of the time-update of the level set Kalman filter}
\label{sec:derive}
In this section, we focus on deriving the \emph{time-update} of the level set Kalman filter (LSKF). In the first subsection, we show that a Gaussian is preserved by a local linear approximation to the original Fokker-Planck equation by tracking its level set. In this process, we observe that the apparent velocity of the level set is given by the drift function plus an additional term which we name as the \textbf{diffusion velocity}. In the second subsection, using the apparent velocity of the level set, we derive a numerical method that tracks such Gaussian particles for the \emph{time-update} step. In the third subsection, we state the \textbf{averaged velocity} version of the \emph{time-update} part of the LSKF, which turns out to give better results numerically.
\subsection{Preservation of Gaussian for a local linear approximation}
\label{ssec:pres_gaus}
Without loss of generality (WLOG), we may assume the particle of concern is centered at $\textbf{0}$. Moreover, since we are interested in how the dynamics and diffusion \emph{deform} the distribution, we may also set the drift function at center $\bv(\mathbf{0}) = \mathbf{0}$. With these simplifications in mind, the \emph{original} Fokker-Planck equation can be restated as: 
\begin{equation}
\label{eqn:heatadv}
\frac{d u}{dt} =  \frac{1}{2}\nabla \cdot \bK \nabla u - \nabla \cdot (\bv u),
\end{equation}
where $u=u(\bx,t)$ is the PDF, $\bK$ is a constant matrix-valued continuous Gaussian \textbf{process noise}, and $\bv = \bv(\bx)$ is the \textbf{drift velocity (field)}, and $\bv(\mathbf{0}) = \mathbf{0}$ by our WLOG simplification.

Then, we approximate (\ref{eqn:heatadv}) by taking a linear approximation of $\bv$: $\bv(\bx) \approx \mathbf{Jx}$. ($\mathbf{J}$ is the Jacobian matrix.) Then: 
\begin{clm}
\label{clm:preserve}
A Gaussian distribution is preserved by a Fokker-Planck equation with a linear drift function. 

Stated explicitly: For the following equation:
\begin{equation}
\label{eqn:approx_diff_eqn}
\frac{d u}{dt} =  \frac{1}{2}\nabla \cdot \bK \nabla u - \nabla \cdot (\mathbf{Jx} u),
\end{equation}
if the initial condition $u(\bx,0)$ is given by a Gaussian function
\begin{equation}
u(\bx,0) = \frac{1}{\sqrt{(2 \pi)^d \det(\bSigma)}} \exp(-\frac{\bx^T \bSigma^{-1}\bx}{2}),
\end{equation}
then $u(\bx,t)$ is also a Gaussian distribution.
\end{clm}
The rest of this subsection proves this claim. 

We define the auxiliary function $F$ by
\begin{equation}
\label{eqn:aux_fun}
F(\bx,t) := \frac{u(\bx,t)}{u(\textbf{0},t)}.
\end{equation}
Consider a \textbf{level set} of the function $F$ at $t$ defined as
\begin{equation}
\label{eqn:levelset}
\mathcal{L}(t) := \left\{\bx \in \mathbb{R}^d | F(\bx,t)  = c \right\},
\end{equation}where $0<c<1$ is some fixed scalar constant.
$\mathcal{L}(t)$ is (usually) a surface, and for a Gaussian as defined in $u(\bx,0)$, it is an ellipsoid. As the function $F$ varies in time, the set $\mathcal{L}$ propagates in space. To describe the movement of the set $\mathcal{L}$, we consider the apparent velocity the traveling surface. 

In particular, a \textbf{velocity of level set} $\bv_\mathcal{L}$ is \emph{defined} by a velocity field satisfying the \textbf{level-set equation}:
\begin{equation}
\label{eqn:level_set_vel}
\frac{dF}{dt}+ \bv_\mathcal{L} \cdot \nabla F = 0.
\end{equation} (Note: this can be understood as the chain rule. For a more detailed explanation, refer to equations (1) and (2) in~\cite{sethian1985curvature}. Also, note the \textbf{velocity of the level set} is uniquely defined up to tangential directions since tangential movements along the level set vanish since they preserve the level set.)

To proceed to the proof, we first consider the lemma:
\begin{lem}
\label{lem:level_set_velocity}
The velocity field 
\begin{equation}
\label{eqn:level_set_ODE}
\bv_L = \bJ\bx + \frac{1}{2} \bK \bSigma^{-1} \bx
\end{equation}	 is a velocity of the level set for $\mathcal{L}(0)$ defined in (\ref{eqn:levelset}). Also, this velocity of the level set is linear in space. 
\end{lem}
\begin{proof}[Proof of Lemma \ref{lem:level_set_velocity}]
First, we note the velocity field is linear in space. To check that it is a velocity of the level set:

Since $F$ is defined as a quotient of $u(\bx,0)$ and $u(\b0,0)$, we may omit the normalizing factor in $u$, and take 
\begin{equation}
\label{eqn:initial_density}
u(\bx,0) =\exp(-\frac{\bx^T \bSigma^{-1}\bx}{2})
\end{equation}
as the initial condition. 
Then: 
\begin{align}
\label{eqn:lem_int_1}
\frac{dF}{dt}|_{t=0} &= \frac{u'(\bx,0)u(\mathbf{0},0) - u'(\mathbf{0},0)u(\bx,0)}{u^2(\mathbf{0},0)}.
\end{align}
By (\ref{eqn:initial_density}), we note that $u(\mathbf{0},0) = 1$. We simplify (\ref{eqn:lem_int_1}) using this substitution:
\begin{align}
\frac{dF}{dt}|_{t=0}&=u'(\bx,0) -  \exp(-\frac{1}{2}\bx^T \bSigma^{-1} \bx)u'(\mathbf{0},0).
\end{align}
Substitute time derivatives $u'$ with (\ref{eqn:approx_diff_eqn}):

\begin{align}
\label{eqn:dFdt_inter_1}
\frac{dF}{dt}|_{t=0} &= \left(\frac{1}{2}\nabla \cdot \bK \nabla u - \nabla \cdot (\bJ\bx u)\right)|_{t=0,\bx=\bx} \\
&- \exp(\dots)\left(\frac{1}{2}\nabla \cdot \bK \nabla u - \nabla \cdot (\bJ\bx u)\right)|_{t=0,\bx=\mathbf{0}} \nonumber,
\end{align}
where
 \begin{equation}
\exp(\dots) := \exp(-\frac{1}{2}\bx^T \bSigma^{-1} \bx).
\end{equation}
continuing the computation, we find that: 
\begin{align}
\nabla \cdot \bK \nabla u &= \nabla \cdot \left(-\bK \exp(\dots) \bSigma^{-1}\bx\right)\\
&= \exp(\dots) \bSigma^{-1}\bx \cdot \bK \bSigma^{-1}\bx - \exp(\dots) \tr(\bK \bSigma^{-1})\\
		&= \exp(\dots)(\bSigma^{-1}\bx \cdot \bK \bSigma^{-1}\bx-\tr(\bK \bSigma^{-1})).
\end{align}
And
\begin{align}
\nabla \cdot((\bJ\bx u))&=\nabla u \cdot \bJ\bx + u \nabla \cdot \bJ\bx\\
&=\exp(\dots)(-\bSigma^{-1}\bx)\cdot \bJ\bx + \exp(\dots)\tr(\bJ)\\
&= \exp(\dots)(\tr(\bJ) - \bSigma^{-1}\bx \cdot \bJ\bx).
\end{align}
Substitute these two terms into (\ref{eqn:dFdt_inter_1}), we have 
\begin{align}
\frac{dF}{dt} &= \exp(\dots)\left(\frac{1}{2}\bSigma^{-1}\bx \cdot \bK \bSigma^{-1}\bx-\frac{1}{2}\tr(\bK \bSigma^{-1}) \right.\\&+\left.\vphantom{\frac{1}{2}\bSigma^{-1}\bx \cdot \bK \bSigma^{-1}\bx} \tr(\bJ) - \bSigma^{-1}\bx \cdot \bJ\bx\right) \nonumber\\
&- \exp(\dots)\left(-\frac{1}{2}\tr(\bK \bSigma^{-1}) +\tr(\bJ)\right)\nonumber\\
\label{eqn:lem_dFdt_final}
&= \exp(-\frac{1}{2}\bx^T \bSigma^{-1} \bx) \left(\frac{1}{2}\bSigma^{-1}\bx \cdot \bK \bSigma^{-1}\bx + \bSigma^{-1}\bx \cdot \bJ\bx\right).
\end{align}
Meanwhile, we check that
\begin{align}
 \bv_L\cdot \nabla F&= \bv_L \cdot \nabla\left(\exp(-\frac{\bx^T \bSigma^{-1}\bx}{2})\right)\\
 &=\bv_L \cdot (-\bSigma^{-1} \bx \exp(-\frac{1}{2}\bx^T \bSigma^{-1} \bx)).
 \end{align}
Substitute the level set velocity term $\bv_L$ from (\ref{eqn:level_set_ODE}), we get that 
 \begin{align}
  \label{eqn:lem_vlnf_final}
 \bv_L\cdot \nabla F&=-\exp(-\frac{1}{2}\bx^T\bSigma^{-1} \bx) (\bJ\bx + \frac{1}{2}\bK \bSigma^{-1} \bx) \cdot \bSigma^{-1}\bx.
\end{align}
Comparing the results from (\ref{eqn:lem_vlnf_final}) and (\ref{eqn:lem_dFdt_final}), we conclude that 
\begin{equation}
\frac{dF}{dt}+ \bv_L \cdot \nabla F = 0.
\end{equation}
Therefore $\bv_L$ is a velocity of the level set, as defined in (\ref{eqn:level_set_vel}).
\end{proof}
We now proceed to the proof Claim \ref{lem:level_set_velocity}:
\begin{proof}[Proof of Claim \ref{clm:preserve}]
Lemma \ref{lem:level_set_velocity} shows that every level set is propagated by a linear velocity field independent of the choice of level set (in other words, independent of the choice of $c$.) In particular, it is propagated by a linear transformation instantaneously. Consequently, between any fixed time $0$ and $\Delta t$, the Gaussian is mapped by a linear transformation. Since a linear transformation maps Gaussians to Gaussians, the velocity field $\bv_L$ is always well-defined as the covariance term in (\ref{eqn:level_set_ODE}) is defined, whereas other terms are known.
\end{proof}
Before we proceed, we should note that the claim is a corollary of equation (29) in~\cite{kalman1961new} by Kalman and Bucy that started the discussion of continuous-discrete Kalman filtering. The significance of the proof is that by using a square-root factorization, the transformation is now given by an \emph{explicit formula} that does not involve an integral. We also note that while the formula (\ref{eqn:level_set_ODE}) is restricted to a normal distribution, tracking distribution by analyzing the propagation of level set as defined in (\ref{eqn:levelset}) can potentially be applied to an $\alpha$-stable distribution introduced in~\cite{samorodnitsky1996stable}. This gives a possibility to extend $\alpha$-stable filtering methods~\cite{leglaive2017alpha,Talebi2018Distributed} to a continuous-discrete problem. 
\subsection{Deriving the time-update of the LSKF}
We now describe the numerical algorithm inspired from the velocity of level set (\ref{eqn:level_set_ODE}). Tracking the movement of the Gaussian is equivalent to tracking one of its ellipsoid level sets (as defined in~(\ref{eqn:levelset})). If the mean of the Gaussian remains at $\mathbf{0}$, then a factorization of the covariance matrix $\bSigma = \bM \bM^T$ can be used to represent the Gaussian. This factorization also represents the unique level set ellipsoid spanned by the columns of the factorization $\bM$. More specifically, set 
\begin{equation}
\bM(0) = \begin{bmatrix}
\bx_1(0)&\cdots&\bx_d(0)\\
\end{bmatrix}
\end{equation}
as initial conditions, and let $\bx_i(t)$ be the solutions of (\ref{eqn:level_set_ODE}). (One may interpret $\bx_i(t)$ as a \textbf{point on the level set}, which travels at the apparent speed defined by a velocity of level set.)
 
Then $\bSigma(t)$ defined as
\begin{equation}
\bSigma(t) := \bM(t) \bM(t)^T
\end{equation}
is the covariance matrix for the Gaussian at time $t$ since it is a similarity transformation. Suppose $A$ is the linear transformation from time $0$ to $t$, then $\bx_i(t) = \mathbf{A} \bx_i(0)$, and 
\begin{align*}
\bSigma(t) &= \bM(t) \bM(t)^T \\
&=\mathbf{A}\bM(0) \bM(0)^T \mathbf{A}\\
&= \mathbf{A} \bSigma(0) \mathbf{A}^T.
\end{align*}
Therefore applying linear transformation $\mathbf{A}$ to the ellipse is equivalent to applying it to all column vectors in $\mathbf{M}$. 

The Jacobian of the velocity field $\bJ$ that appears in (\ref{eqn:level_set_ODE}) is not explicitly needed. Instead of direct evaluation of the Jacobian, we approximate the effect of the drift velocity by applying a quadrature rule. In this section, we use the forward difference in space to derive the method. Specifically, we notice $\bSigma^{-1} = \bM^{-T}\bM^{-1}$, (where $\bM^{-T}$ indicates the inverse transpose) and the $\bx_i$s are columns of $\bM$. Therefore for all $\bx_i$s on the level set ellipsoid, by substituting the terms in (\ref{eqn:level_set_ODE}), we find the apparent velocity of the level set is approximated by
\begin{equation}
\label{eqn:level_set_ODE_num}
\frac{d \bx_i}{dt} = \mathbf{\bv}(\bar{\bx}+\bx_i) - \bv(\bar{\bx}) + \frac{1}{2}\bK(\bM^{T})^{-1}\mathbf{e}_i,
\end{equation}
where $\bar{\bx}$ is the \textbf{mean} of the Gaussian. In the case the Gaussian is centered at $\mathbf{0}$, $\bar{\bx} = \mathbf{0}$. $\mathbf{e}_i$ is the $i$th  unit vector with all entries $0$ except that $i$th entry is $1$. 

Recall that the $\bx_i$s are columns of $\bM$. In a matrix short-hand (where the matrix-vector additions are defined entry-wise, and recall vectors are column vectors): 
\label{sec:level_set_plain}
\begin{equation}
\label{eqn:level_set_ODE_num_matform}
\frac{d \bM}{dt} = \bv(\bar{\bx}+\bM) - \bv(\bar{\bx}) + \frac{1}{2}\bK(\bM^{T})^{-1}.
\end{equation}
Whereas the velocity for center is given by 
\begin{equation}
\frac{d \bar{\bx}}{dt} = \bv(\bar{\bx}).
\end{equation}
Based on the form of the equations, the assumptions $\bar{\bx} = \mathbf{0}$ and $\bv(\bar{\bx}) = \mathbf{0}$ can be dropped. 

Concatenating $\bar{\bx}$ and $\bM$ as a variable $(\bar{\bx}|\bM)$ of dimension $d\times (d+1)$, we obtained a nonlinear ODE in this space. Any standard ODE solver can be applied to this ODE to complete the \emph{time-update} between the measurements.
 
Note that to evaluate the velocity of one point $\bx_i$ on the level set, both the mean $\bar{\bx}$ and all other points $\bx_j$ for this Gaussian kernel are needed, hence the points on a level set cannot be updated independently (in contrast to the time-update step for both the UKF and the CD-CKF). \emph{This provides intuition about the difference between our method and others: while other methods looks at the past covariance information and rely on an expansion in time, our method uses only the current information about the covariance matrix. Except for the purpose of numerically solving the ODE, our method does not need time-discretization.}
\subsection{Motivating example: linear drift function}
As an illustration for the \emph{time-update} method, we consider the following Fokker-Planck equation with a linear drift function:
\begin{equation}
\label{eqn:linear_diff_example}
\frac{d u}{dt} =  \nabla \cdot \bK \nabla u - \nabla \cdot (\mathbf{Jx} u)
\end{equation}
with parameters
\begin{equation*}
\begin{matrix}
\bK = \begin{bmatrix}
\frac{1}{2}&\frac{1}{4}\\
\frac{1}{4}&\frac{3}{2}
\end{bmatrix}&
\mathbf{J} = 
\begin{bmatrix}
0&0.1\\
0&0
\end{bmatrix}.
\quad

\end{matrix}
\end{equation*}
We consider the solution of the initial value problem with initial condition
\begin{equation}
u(\bx,0) = \frac{1}{\sqrt{(2 \pi)^d \det(\bSigma_0)}} \exp(-\frac{\bx^T \bSigma_0^{-1}\bx}{2}),
\end{equation}
where the initial covariance is given by
\begin{equation}
\bSigma_0 = \begin{bmatrix}
2&1\\
1&2
\end{bmatrix}.
\end{equation}

In Sec.~\ref{sec:derive}, we proved that propagating level sets by (\ref{eqn:level_set_ODE}), and consequently (\ref{eqn:level_set_ODE_num_matform}) exactly solves~(\ref{eqn:linear_diff_example}). 
To give a concrete numerical example of this property, we check the convergence of numerical ODE solvers for the initial value problem and find the error of the density function at $t=10$ with ODE solvers of different order. We factor $\bSigma_0 = \bM_0 \bM_0^T$, and set the ODE (\ref{eqn:level_set_ODE_num_matform}) with the initial condition $\bM(0) := \bM_0$. To verify that our method is accurate for the linear Fokker-Planck equation (\ref{eqn:linear_diff_example}), we check that when using different numerical ODE solvers, the solution converges to the same value, with the rate of convergence coinciding with the order of the ODE solver.

\begin{figure}
\includegraphics[width = 0.5\textwidth]{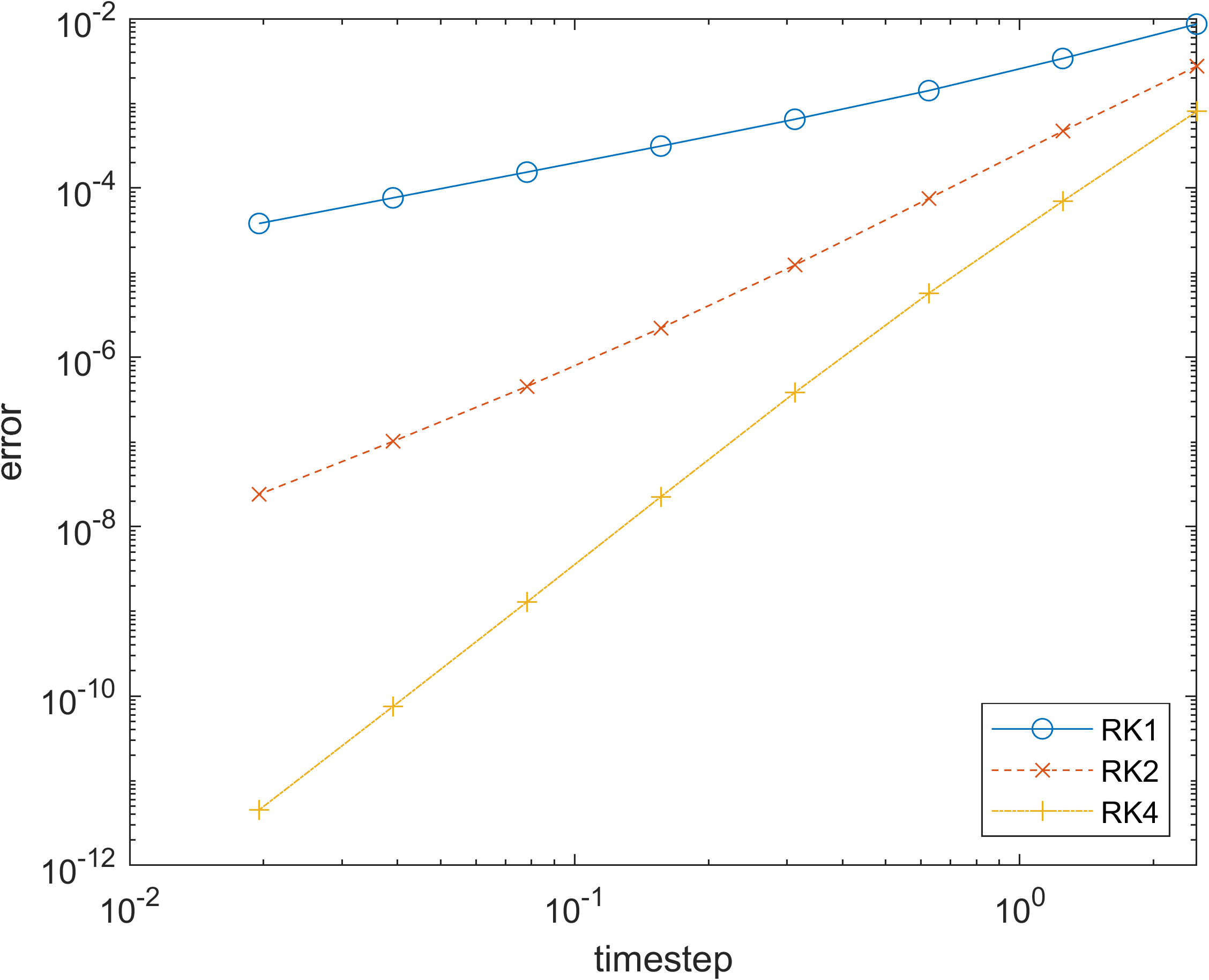}
\caption{log-log graph of the error of covariance matrix in the infinity norm with linear Fokker-Planck equation. This shows that our method preserves the order of accuracy of the ODE solver.}
\label{fig:linear_conv}
\end{figure}
In Fig. \ref{fig:linear_conv}, we verify that for the Runge-Kutta methods of order $1$, $2$, and $4$, the error of $u(\cdot,10)$ measured in infinity norm converges at the same order of the ODE solvers, which is expected if the reformulation (\ref{eqn:level_set_ODE_num_matform}) is exact for (\ref{eqn:linear_diff_example}). 

\subsection{The averaged velocity level set time-update}
Here we state the \textbf{averaged velocity level set time-update} method, which uses central difference instead of forward difference, and shows better accuracy in numerical experiments (See Appendix A for an example) when compared versus the version in Sec. \ref{sec:level_set_plain}.

We set the velocity of the mean $(d \bar{\bx})/dt$ by the \textbf{averaged velocity}:
\begin{equation}
\label{eqn:center_ODE_avg}
\frac{d \bar{\bx}}{dt} = \bv_a(\bar{\bx},\bM) := \frac{1}{2d} \sum_{i=1}^d\left(\bv(\bar{\bx} + \bx_i) + \bv(\bar{\bx} -\bx_i)\right).
\end{equation} (Recall that $\bx_i$ are columns of the matrix $\bM$)

and the velocity of the matrix $\bM$:
 \begin{equation}
 \label{eqn:level_set_ODE_num_matform_avg}
\frac{d \bM}{dt} = \mathbf{\bv}(\bar{\bx}+\bM) - \mathbf{v_a} + \frac{1}{2}\bK(\bM^{T})^{-1}.
\end{equation}
It can be easily seen that when the drift velocity field $\bv$ is linear in space, equations (\ref{eqn:level_set_ODE_num_matform_avg}) and (\ref{eqn:level_set_ODE_num_matform}) are identical, and  $\bv_a(\bar{\bx},\bM)  = \bv(\bar{\bx})$. 

Using this averaged velocity, here we summarize the LSKF:

\begin{algorithmic}[1]
\REQUIRE{Guess of initial state $\hat{\bx}_0$ at time $t_0$, and a factorization of a guess of covariance matrix $\hat{\bM}_0$, measurements $\by_1 \dots,\by_n$ at time $t_1,\dots,t_n$.

Drift velocity $\bv$, continuous process noise matrix $\bK$, measurement function $h$, a factorization of the covariance matrix $\mathbf{R}$ of a zero-mean Gaussian measurement noise.}
\FOR{j=1,\dots,n}
\STATE{Set the problem of $\bar{\bx}(t)$ and $\bM(t)$ given by:
\begin{align}
\frac{d\bar{\bx}}{dt} &= \frac{1}{2d} \sum_{i=1}^d(\bv(\bar{\bx} + \bx_i) + \bv(\bar{\bx} -\bx_i))\\
\frac{d \bM}{dt} &= \mathbf{\bv}(\bar{\bx}+\bM) - \mathbf{v_a} + \frac{1}{2}\bK(\bM^{T})^{-1},
\end{align}
with initial condition $\bar{\bx}(t_{j-1}) = \hat{\bx}_{j-1}$, and $\mathbf{M}(t_{j-1}) = \hat\bM_{j-1}$
(Recall: $\mathbf{v_a}$ is defined in (\ref{eqn:center_ODE_avg}), also recall that $\bx_i$ are columns of $\bM$.)
 }
 \STATE{\emph{Time-update}: solve the above equation from $t_{j-1}$ to $t_j$ using a numerical ODE solver, and approximate $\bar{\bx}_j = \bar{\bx}(t_j)$, $\bM_j = \bM(t_j)$ with the numerical solution.}
\STATE{\emph{Measurement-update}: Find the corrected mean $\hat{\bx}_j$ and corrected covariance matrix $\hat{\bM}_j$ at time $t_j$ by applying the measurement-update algorithm defined in subsection \ref{ssec:sqcdckf}, with input $\bar{\bx}_j$, $\bM_j$, and $\mathbf{R}$.}
\ENDFOR

\RETURN{Predicted corrected state $\hat{\bx}_1,\hat{\bx}_n$, at $t_1,\dots,t_n$, with a factorization of the predicted corrected covariance matrix $\hat{\bM}_1,\dots,\hat{\bM}_n$.}
\end{algorithmic}
\subsection{Comparing convergence: achieving beyond IT-1.5 without explicit higher derivatives}
In Sec. \ref{ssec:CD-CKF}, we introduced the CD-CKF with IT-1.5. Here, we compare the convergence rate of the time-update of the CD-CKF (as implemented in~\cite{arasaratnam2010cubature}, and with proper IT-1.5) with that of the LSKF.

Consider a simple harmonic oscillator:
\begin{equation}
\bx(t) := [\epsilon(t)\quad\dot{\epsilon}(t)\quad\ddot{\epsilon}(t)]^T,
\end{equation}
where $\epsilon$, $\dot{\epsilon}$, and $\ddot{\epsilon}$ are the position, velocity, and acceleration of the oscillator. Its time derivative is given by
\begin{equation}
\bv(\bx) = [\dot{\epsilon}\quad\ddot{\epsilon}\quad-\epsilon]^T.
\end{equation}
This oscillator is also subject to a continuous process noise, defined by the diagonal diffusion matrix:
\begin{equation}
\mathbf{K} = \text{diag}[0.01^2\quad0.01^2\quad0.02^2].
\end{equation}
Since the dynamics are linear, the Gaussian is preserved, and we expect the result from the LSKF and the proper IT-1.5 to converge to the exact solution.

To find the order of convergence, and compare the methods, we consider the following initial condition problem. Given initial condition
\begin{equation}
\bx(0) = [1\quad0\quad0]^T
\end{equation} and initial covariance matrix
\begin{equation}
\mathbf{\Sigma}(0) = \text{diag}[0.01^2\quad0.01^2\quad0.03^2],
\end{equation}
we would like to find the end state at $t=0.2$ using the above mentioned methods. By subdividing the timesteps, we arrive at the following convergence result:
\begin{figure}
\includegraphics[width=0.5 \textwidth]{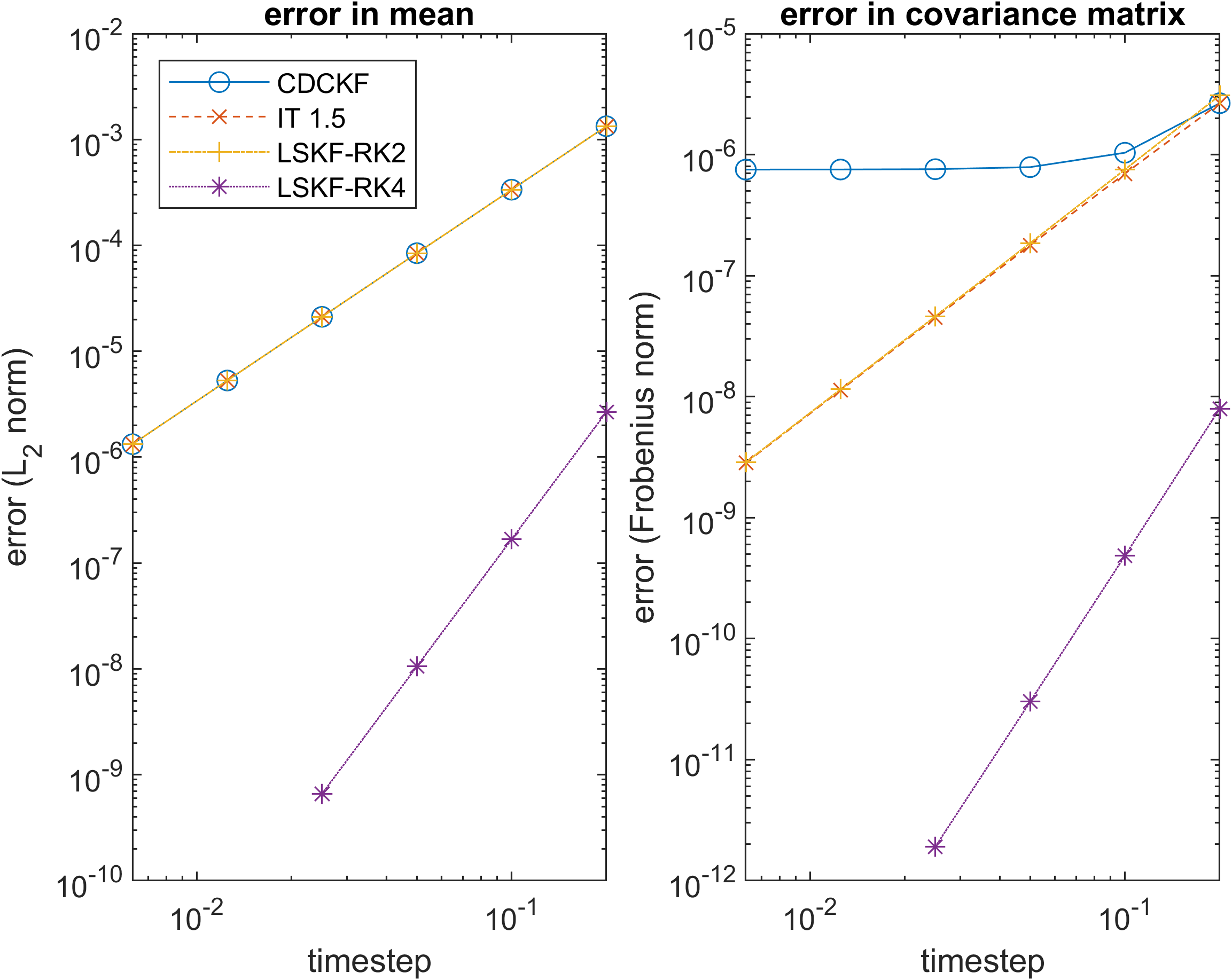}
\caption{convergence of the mean $\bx$ and covariance matrix with the CD-CKF and the LSKF. The left panel shows the error (measured in $L_2$ norm) in mean value as a function of the timestep, whereas the right panel shows error (measured in Frobenius norm) in the covariance matrix. In each panel, CD-CKF~(blue) denotes the time-update implemented in~\cite{arasaratnam2010cubature}, IT-1.5~(red) denotes the CD-CKF with the proper IT-1.5 expansion, LSKF-RK2~(yellow), and LSKF-RK4~(purple) denotes time-update of the LSKF with Runge-Kutta solvers of order $2$ and $4$ respectively. The Runge-Kutta 4 version is cut short due to finite precision linear algebra. }
\label{fig:linear_compare}
\end{figure}
Similar to Fig.~\ref{fig:linear_conv}, Fig.~\ref{fig:linear_compare} shows that the time-update of the LSKF converges the same order as the underlying ODE solver. Importantly, note the proper IT-1.5 and the LSKF-RK2 converge to the same limit mean and covariance matrix at a weak order of convergence 2, validating the correctness of both methods. (Note: they \textbf{do not} converge to the same square root of the covariance matrix, which is not surprising given that the matrix square roots are not unique.) The time-update of the CD-CKF as implemented in~\cite{arasaratnam2010cubature} does not converge to the same limit. Comparing the LSKF-RK4 versus the proper IT-1.5, it can be noted that much fewer timestep subdivisions can achieve similar truncation errors. As pointed out in~\cite{sarkka2012continuous}, the IT-1.5 already requires explicit first and second derivatives of $\bv(\bx)$, and any higher-order Ito-Taylor expansion is necessarily more complicated. On the contrary, the time-update of the LSKF, as defined in (\ref{eqn:level_set_ODE_num_matform_avg}) does not require the explicit expression of the derivatives of $\bv(\bx)$, and can achieve a higher order of convergence with the freedom to choose any ODE solver.
\section{Numerical example: the radar tracking coordinated turn test case}
\subsection{Problem description}

Here, we follow the test case presented in~\cite{arasaratnam2010cubature}, considering the scenario where a radar station tracks an aircraft making a coordinated turn. Since the CD-CKF in~\cite{arasaratnam2010cubature} is claimed to be \emph{the choice for challenging radar problems}, we compare LSKF against CD-CKF in the most challenging scenario they considered with $\omega = 6^\circ/s$ with sampling intervals $T = 2s$, $T=4s$, and $T = 6s$. (Note: conversion to radians per second is required) Additionally, we consider the more challenging scenario with $\omega = 12^\circ/s$ and $\omega = 24^\circ/s$. We implemented the CD-CKF based on the square root form formulated in~\cite{arasaratnam2010cubature}, using their implementation presented at~\cite{arasaratnamCKFcode}. The details of the test case are as follows:

The aircraft is described by a $7$-dimensional state vector 
\begin{equation}
\mathbf{x}(t) := [\epsilon(t) \quad \dot{\epsilon}(t) \quad\eta(t) \quad \dot{\eta}(t)\quad \zeta(t) \quad \dot{\zeta}(t) \quad \omega(t)]^T,
\end{equation}
where $\epsilon(t),\eta(t),\zeta(t)$ describes the position, in meters, $\dot{\epsilon}(t), \dot{\eta}(t), \dot{\zeta}(t)$ describes the velocity of the aircraft, in meters per second, and the $\omega(t)$ describes the turn rate of the aircraft, in \textbf{radians} per second. The dynamics of the aircraft are defined by the following drift equation:
\begin{equation}
\mathbf{v}(\mathbf{x}(t)) = [\dot{\epsilon} \quad -\omega \dot{\eta} \quad \dot{\eta} \quad \omega \dot{\epsilon} \quad \dot{\zeta} \quad 0 \quad 0]^T.
\end{equation}
The noise term is defined by the following diagonal diffusion matrix:
\begin{equation}
\mathbf{K} = \text{diag}([0\quad \sigma_1^2 \quad 0 \quad \sigma_1^2 \quad 0 \quad \sigma_1^2 \quad \sigma_2^2]),
\end{equation}
where $\sigma_1 = \sqrt{0.2}$, and $\sigma_2 = 7\times10^{-4}$. (Note: in~\cite{arasaratnam2010cubature}, they suggested $\sigma_2 = 7\times10^{-3}$. However, the accompanied code provided by Arasaratnam on his webpage~\cite{arasaratnamCKFcode} used the parameter $\sigma_2^2 = 5\times10^{-7}$, which matches closely with $\sigma_2 = 7\times10^{-4}$. Our calculated RMSE also turns to be similar as shown in Fig. 2,3 and 4 in~\cite{arasaratnam2010cubature} if $7\times10^{-4}$ is chosen, whereas $\sigma_2 = 7\times10^{-3}$ does not give similar results.)

The measurement is from a single radar station located at $\mathbf{s} = [1500\quad10\quad0]$. The radar station measures the distance $r$, azimuth angle $\theta$ and elevation angle $\phi$ relative to the radar station. The measurement function is therefore given by:
\begin{equation}
\begin{bmatrix}
r\\
\theta\\
\phi
\end{bmatrix}=
\begin{bmatrix}
\sqrt{(\epsilon-1500)^2+(\eta-10)^2 + \zeta^2}\\
\arctan(\frac{\eta-10}{\epsilon-1500})\\
\arctan(\frac{\zeta}{\sqrt{(\epsilon-1500)^2+(\eta-10)^2}})
\end{bmatrix} + \mathbf{w}.
\end{equation}
where the measurement noise $\mathbf{\tau} \sim \mathcal{N}(0,\mathbf{R})$, with measurement noise matrix $\mathbf{R} = \text{diag}([\sigma_r^2,\sigma_\theta^2,\sigma_\phi^2])$, where $\sigma_r = 50, \sigma_\theta = 0.1^\circ, \sigma_\phi = 0.1^\circ$ (Note: the standard deviations $\sigma_\theta$ and $\sigma_\phi$ are measured in \textbf{degrees}, and a unit conversion is needed).

For the test scenario, the aircraft starts with the initial state 
\begin{equation}
\mathbf{x}_0 = [1000\quad 0 \quad 2650\quad 150 \quad 200\quad 0\quad \omega_0 ]^T,
\end{equation}
where $\omega_0$ is the initial turn rate, and the measurement is taken with a constant time interval $T$. The total time for simulation is chosen to be $120$ seconds. The turn rate and measurement interval vary across test cases to examine the performance of the filters. The initial covariance is taken as
\begin{equation}
\mathbf{\Sigma} = \text{diag}([100\quad 1\quad 100 \quad1\quad 100 \quad1\quad 0.01]),
\end{equation}
based on a physically realistic assumption: from an observer on the ground, one would have a reasonably good guess about its position with standard deviation $\sigma = 10$ meters, and a good guess about its velocity through differentiation with $\sigma = 1$ meters per second, but a rather bad guess for the turn rate with $\sigma = 0.1$ radian per second, or approximately $5.73$ degrees per second.

\subsection{Numerical results}
\label{ssec:numerical_over}
With the problem description complete, we now turn to present our numerical results. $N =100$ experiments are executed for each set of parameters, and the same set of experiments is applied to all candidate filters. The main performance metric used is the Root-mean square error (RMSE) for position, velocity and turn rate. For example,
RMSE for position is defined as:
\begin{equation}
\sqrt{\frac{1}{NK}\sum_{n=1}^{N}\sum_{k=1}^{K}\left((\epsilon_k^n-\hat{\epsilon}_k^n)^2+ (\eta_k^n-\hat{\eta}_k^n)^2+(\zeta_k^n-\hat{\zeta}_k^n)^2\right)},
\end{equation}
 where $N$ is the number of experiments, and $K$ is the number of measurements in each experiment. 
 
Another metric we consider is the number of divergent results, which we define as any result that has an error larger than $500$ or ends prematurely due to a \texttt{not a number} error. 
Following~\cite{arasaratnam2010cubature}, we evaluate the performance of each method at different subdivisions, $m$, of the timestep between measurements. Since our method is defined purely as a reformulated ODE, the subdivision of the timestep is the same as a timestep in a fixed timestep ODE solver, such as the widely-used Runge-Kutta 4 method. Additionally, to verify our claim that the choice of timestep subdivision can be completely passed to the ODE solver, we also use an adaptive solver, \texttt{ode113}, which is integrated into the \texttt{MATLAB} software package. In this sense, we introduced $2$ implementations of the LSKF, which we call the \emph{LKSF-RK4} and \emph{LSKF-adaptive} respectively. In the following examples, we will verify that the additional subdivisions do not affect the numerical results from the LSKF-adaptive. 
\begin{figure*}[t]

\includegraphics[width=\textwidth]{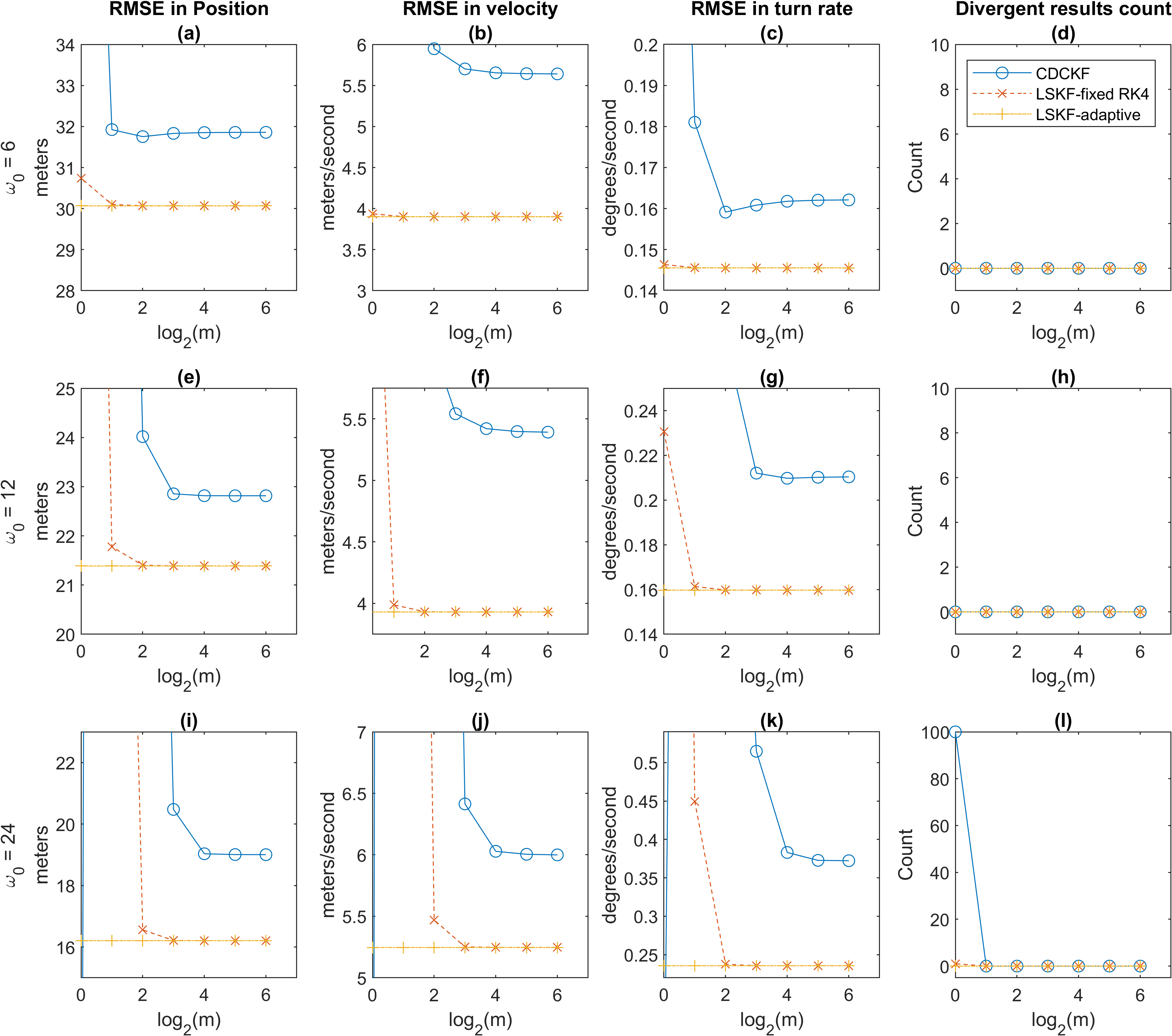}
\caption{RMSE and count of divergence results, for a fixed measurement interval $T = 6s$ varying $m$ and $\omega_0$, where $m$ is the number of timestep subdivisions between measurements, and $\omega_0$ is the initial turn rate. Each row gives the performance metrics for the same initial turn rate, whereas each column contrasts the same performance measurement across different initual turn rates. Note for $\omega_0=24^\circ/s$ and $m=1$, all results from the CD-CKF are divergent.}
\label{fig:RMSE6sec}
\end{figure*}

Our numerical results in Fig. \ref{fig:RMSE6sec} show that our methods consistently outperform the CD-CKF in this test case, across all choices of angular velocity and timestep subdivisions. Importantly, note that the performance of the CD-CKF cannot match that of the LSKF-adaptive even if sufficient timestep subdivisions are introduced. We suspect this is due to the fact that the CD-CKF as introduced in~\cite{arasaratnam2010cubature} only uses the IT-1.5 expansion at the beginning of each \emph{time-update} step but not at the subdivided timesteps, whereas our method is defined using instantaneous information and is not subject to this limitation. 

Equally importantly, note that the LSKF-adaptive version of our method gives the same result independent of the subdivisions introduced. Additionally, the fixed-timestep LSKF-RK4 converges to the LSKF-adaptive result, as expected for a consistent ODE solver. In practice, a user can always use the LSKF-adaptive version with the choice of adaptive ODE solvers that gives the best performance without needing to consider timestep subdivisions manually.
\begin{figure*}[t]
\includegraphics[width=\textwidth]{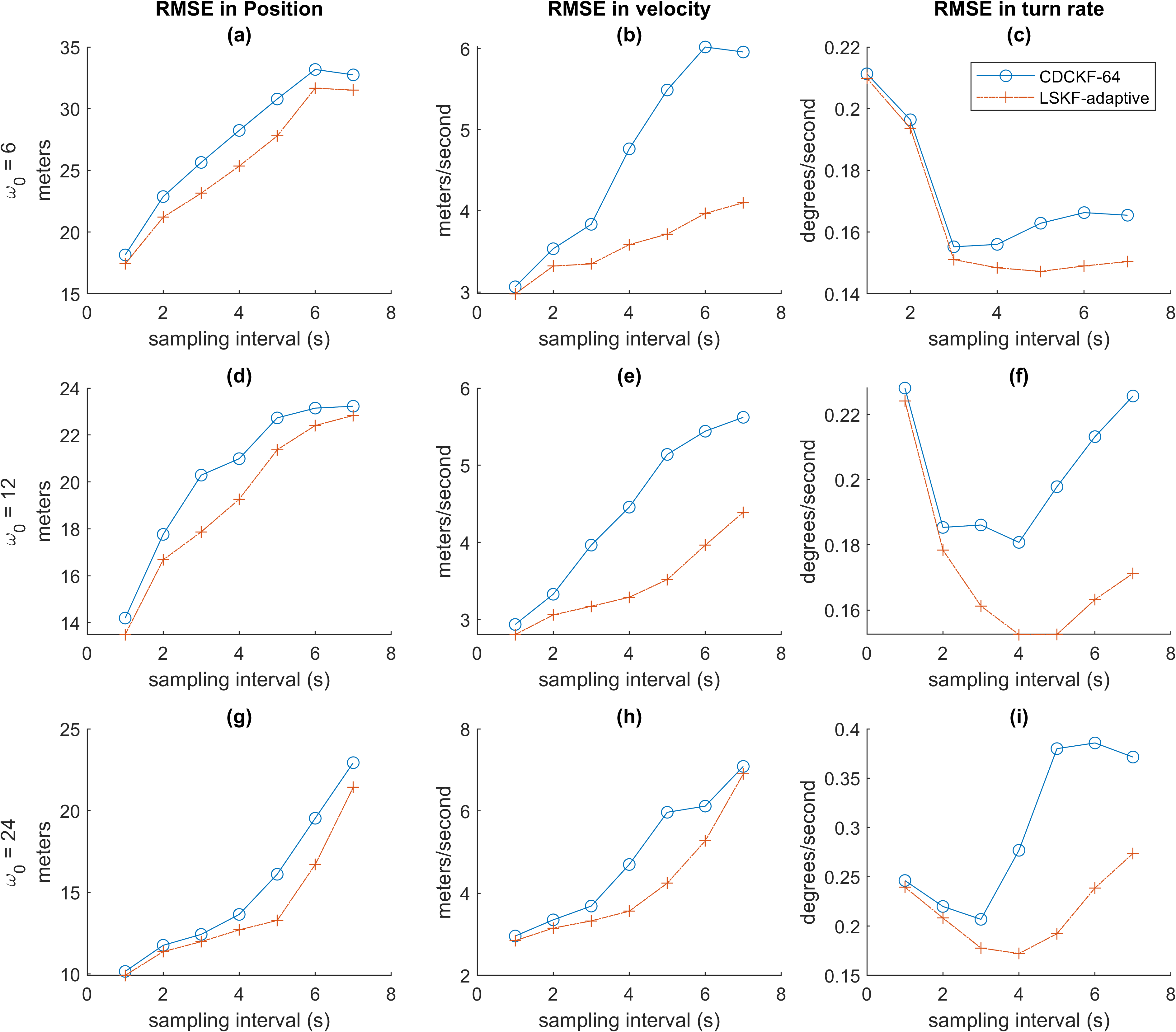}
\caption{RMSE with varying the measurement interval from $1$ through $7$ seconds, and varying the initial turn rate $\omega_0$. For the CD-CKF, a timestep subdivision of $m=64$ is chosen to ensure that the CD-CKF is performing optimally. For the LSKF, an adaptive ODE solver is used and no timestep subdivisions are manually inserted. All results from the CD-CKF and the LSKF are convergent.}
\label{fig:RMSE_var_sec}
\end{figure*}
In Fig.~\ref{fig:RMSE_var_sec}, we verify that even with sufficient timestep subdivision for the CD-CKF, the LSKF still outperforms the CD-CKF over all the parameters chosen, even when no intermediate timesteps are manually inserted. Additionally, the difference in performance between the CD-CKF and the LSKF is more significant when the measurement interval $T$ is large, whereas the results are similar when $T=1s$. 
With this in mind, we proceed further into the numerical experiments, using sufficient timestep subdivisions ($m=64$) for the CD-CKF, whereas no additional pre-defined timestep subdivisions ($m=1$) for the LSKF-adaptive, and vary the measurement interval $T$ from $1$ second to $7$ seconds with an increment of $1$ second. 

In conclusion, for the test case picked by~\cite{arasaratnam2010cubature}, our LSKF method consistently outperforms the CD-CKF across the challenging scenarios introduced by them. In addition, our method requires less input from an end-user, as our method only requires knowledge of the \emph{drift function} explicitly, whereas the CD-CKF also requires the first and second spatial derivatives of the \emph{drift function}, as well as a user-defined timestep subdivision parameter $m$. Finally, the elegance of reforming the system as an ODE without introducing expansion in time gives more room for possible future improvement.



\section{Conclusion}
In this paper, we derived a novel Level Set Kalman Filter method for nonlinear continuous-discrete systems. From a theory standpoint, our derivation is based on the movement of a level set instead of the moments of a distribution. Our method reformulates the \emph{time-update} of the filtering as an ODE. As a consequence of this formulation, our description is instantaneous, in contrast to existing methods that use some expansion in time to approximate the continuous process noise. 

From a practical point, for the radar tracking coordinated turn test case, our method consistently outperforms the CD-CKF over a range of challenging scenarios. Additionally, our method requires less explicit information about the model, and our instantaneous formulation allows a user to easily pass the task of choosing a timestep to the well-established field of adaptive ODE solvers. The numerical results indicate that our method is a good candidate for challenging tracking problems, especially if an appropriate timestep cannot be determined \emph{a priori}.

\section*{Acknowledgment}
The authors acknowledge support from National Science Foundation through grant NSF DMS-1714094. D. B. Forger is the CSO and holds equity in Arcascope. We especially thank the anonymous reviewers and the associated editor, whose advice helped us improve the numerical stability of our method, in addition to the general improvement in the structure of the manuscript. 
\section*{Appendix A: Comparison of the standard, averaged and partially averaged time-update of the LSKF}
In equations (\ref{eqn:level_set_ODE_num_matform}) and (\ref{eqn:level_set_ODE_num_matform_avg}), we defined the  \textbf{(standard)} \emph{time-update} and the \textbf{averaged velocity} \emph{time-update} equations. Here we use a numerical example to illustrate the difference in behavior between these methods. Additionally, we introduce the \textbf{partially averaged velocity} \emph{time-update} equations as a trade-off option between the standard and the averaged velocity version. 

Using the same notations as in (\ref{eqn:level_set_ODE_num_matform}) and (\ref{eqn:level_set_ODE_num_matform_avg}), the \textbf{partially averaged velocity} is defined as
\begin{equation}
\mathbf{v_p}:= \frac{1}{2d} \left(\sum_{i=1}^d(\bv(\bx + \bx_i)) + d \times \bv(\bx -\bx_1)\right).
\end{equation}
Correspondingly, in matrix form, the \emph{time-update} ODE using partially averaged velocity is defined by
 \begin{equation}
 \label{eqn:level_set_ODE_num_matform_pavg}
\frac{d \bM}{dt} = \mathbf{\bv}(\bx+\bM) - \mathbf{v_a} + \frac{1}{2}\bK(\bM^{T})^{-1}.
\end{equation}
As an illustrative example to show the effects of using an \emph{averaged velocity} or \emph{partially averaged velocity} versus evaluating the velocity at the center when a nonlinear drift velocity is present, we consider the following system: 
\begin{equation*}
\frac{du}{dt} = -\nabla \cdot (\bv u),
\end{equation*}
where $a,b$ are positive constants, and
\begin{align*}
    \bv(x,y,0) &= (0,x^2)\\
    u(x,y,0) &= \frac{1}{2 \pi ab} \exp(-\frac{\frac{x^2}{a^2}+\frac{y^2}{b^2}}{2}).
\end{align*}
(Note $x,y$ in the following equation are not in bold, and are scalars)
The analytic solution to this transport equation is given by 
\begin{equation*}
    u(x,y,0) = \frac{1}{2 \pi ab} \exp(-\frac{\frac{x^2}{a^2}+\frac{(y-xt)^2}{b^2}}{2}).
\end{equation*}
To compare the performance of the three LSKF methods, we compute the averaged $\mathbf{L^2}$ error of the results with the analytical result, with randomly chosen matrix square root $\mathbf{M}\mathbf{M}^T = \bSigma$. 
\begin{figure*}[!t]
        \centering
             \subfloat[]{\includegraphics[width=0.3\textwidth]{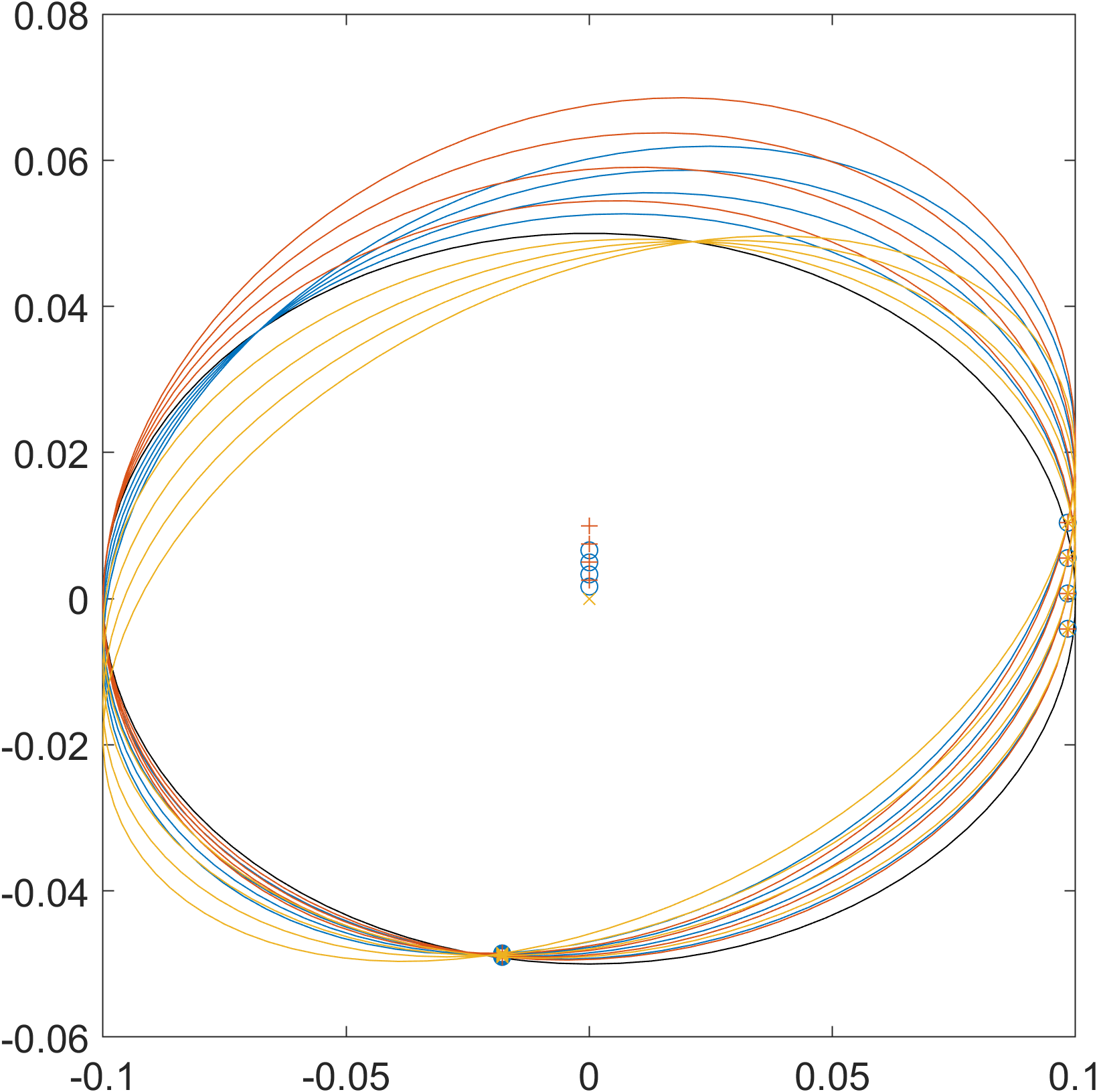}}  
\hfil
             \subfloat[]{\includegraphics[width=0.3\textwidth]{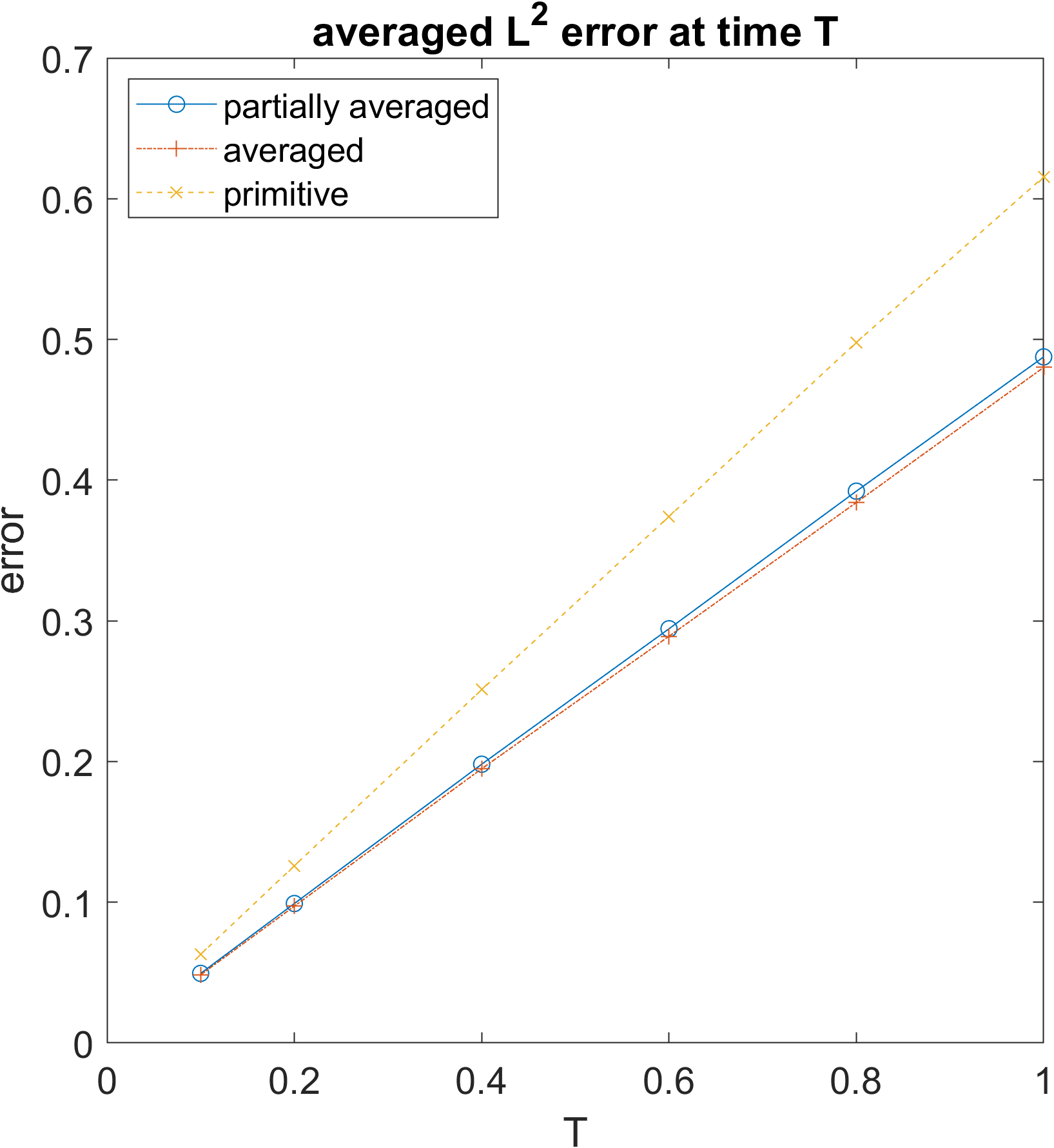}}  
             \hfil
             \subfloat[]{\includegraphics[width=0.3\textwidth]{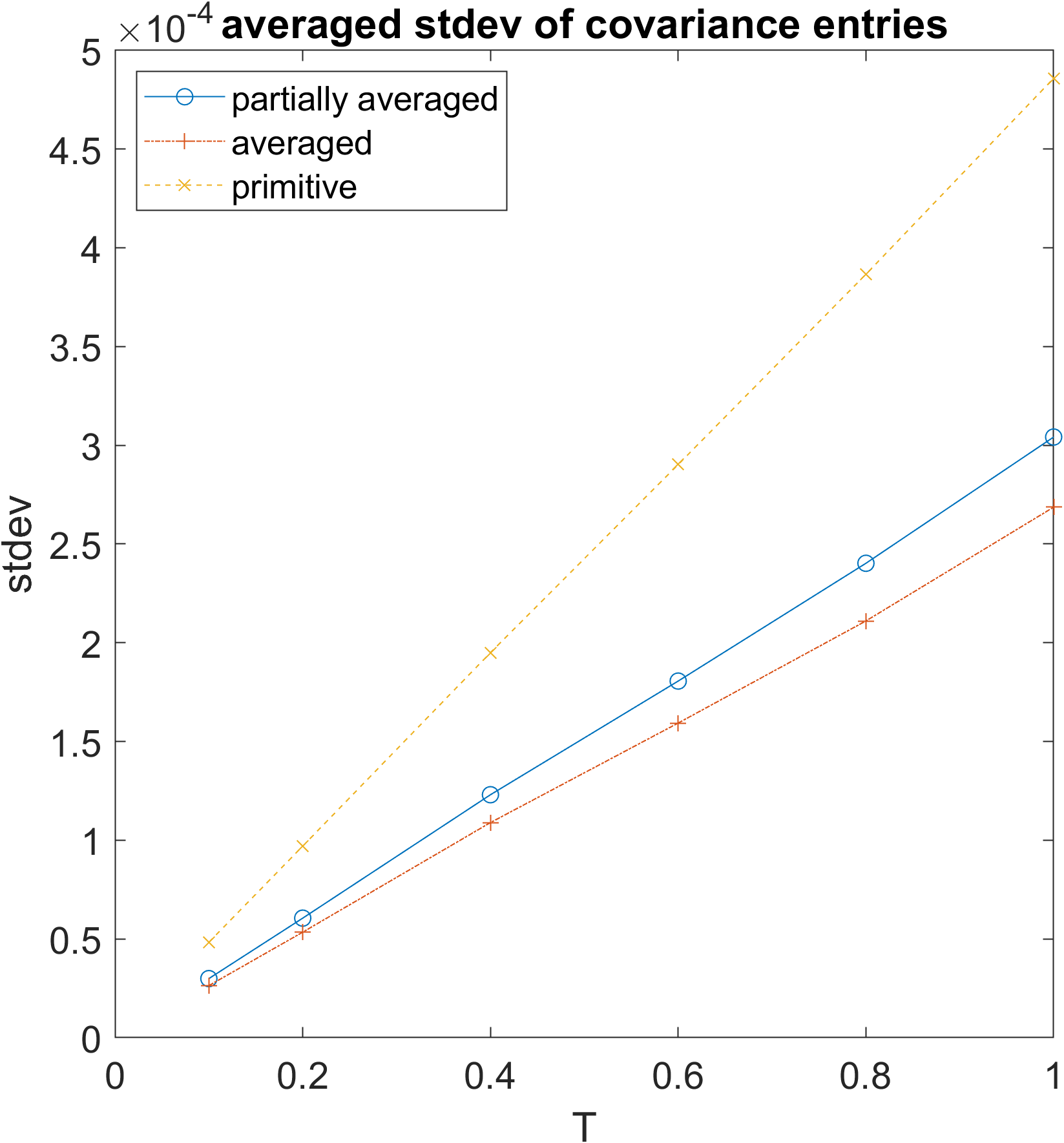}}
        
        \caption{A comparison between the partially averaged (blue), averaged (red), and primitive (yellow) LSKF applied to a nonlinear system. (a) shows the trajectories of the center and level set. (b) shows the averaged $\mathbf{L^2}$-norm of error over $1024$ trials. (c) shows the averaged standard deviation of the covariance matrix.}
        \label{fig:pavg_avg_prim}
    \end{figure*}
As can be observed from Fig.~\ref{fig:pavg_avg_prim}, the partially averaged velocity and averaged velocity version has less RMSE than the standard method. The averaged velocity method is less sensitive than the partially averaged method in that the covariance entries are less dependent on the choice of matrix square root. However, since the partially averaged method requires $d+1$ evaluations of the drift velocity whereas averaged method requires $2d$ evaluations, the partially averaged method is still useful as it can be considered as an efficient improvement from the standard version.

\section*{Appendix B: Computational cost of the LSKF}
A count of FLOPs would be misleading for this method, as the main function reformulates the function as an ordinary differential equation, and the number of steps used is highly dependent on the choice of the numerical ODE solver and the numerical properties of the problem when an adaptive solver is used. When using an ODE solver, most of the computational cost is associated with evaluating the derivative. Therefore, we find the number of evaluation of the \textbf{drift velocity} $\bv$, and FLOPs needed for a single derivative evaluation in (\ref{eqn:level_set_ODE_num_matform_avg}). The computation cost for the \emph{time-update} is then mainly decided by the number of derivative evaluations needed and the length of the timestep. The \emph{measurement-update} is identical to that of the CD-CKF, which is listed in Table V of~\cite{arasaratnam2010cubature} and is omitted here. 

Computations needed for (\ref{eqn:level_set_ODE_num_matform_avg}) are:
\begin{enumerate}
\item
Evaluate the \textbf{averaged velocity}:
\begin{equation}
\label{eqn:center_ODE_avg}
\frac{d \bar{\bx}}{dt} = \bv_a(\bar{\bx},\bM) := \frac{1}{2d} \sum_{i=1}^d(\bv(\bar{\bx} + \bx_i) + \bv(\bar{\bx} -\bx_i)).
\end{equation} $2d$ \textbf{drift velocity} evaluations and $O(d^2)$ FLOPs. 
\item
Find $\mathbf{K}(\mathbf{M}^T)^{-1}$: if solved by LU factorization, forward substitution, and back substitution: $\frac{8}{3}d^3 + O(d^2)$ FLOPs.
\item
Evaluate (\ref{eqn:level_set_ODE_num_matform_avg}): $O(d^2)$ FLOPs.
\end{enumerate}
In total, $2d$ \textbf{drift velocity} evaluations and $\frac{8}{3}d^3+O(d^2)$ FLOPs are needed for evaluating (\ref{eqn:level_set_ODE_num_matform_avg}). When using a fixed Runge-Kutta 4 method, $4$ such evaluations are needed per timestep, which results in $8d$ \textbf{drift velocity} evaluations and $\frac{32}{3}d^3 + O(d^2)$ FLOPs per \emph{time-update}. If the measurement interval is short, then a fixed Runge-Kutta 2 method can be used, with half evaluations needed. For most scenarios, we suggest using an adaptive solver to automatically determine the appropriate timestep given a target error bound.
\begin{table}[h]
\caption{List of symbols and notations}
\begin{tabularx}{0.5 \textwidth}{cX}
\hline
$\mathbf{\tau}$&measurement noise\\
$\mathcal{N}(\mu,\mathbf{R})$& multivariate normal distribution with mean $\mu$ and covariance $\mathbf{R}$\\
$\mathbf{R}$&covariance of the measurement noise\\
$d$&dimension of the state vector\\
$\mathbf{x}$&state vector\\
$\mathbf{y}$&measurement vector\\
$h(\cdot)$&measurement function\\
$\mathbf{v}$&the (drift) velocity of the system dynamics\\
$\Delta t$&the timestep taken with the \emph{time-update} step of the CD-CKF and the LSKF\\
$\mathbf{f}_d(\bx,t)$&a function associated with the \emph{time-update} step of the CD-CKF with a timestep of $\Delta t$\\
$\mathbf{K}$&$d\times d$ covariance matrix of the continuous process noise\\
$\sqrt{\mathbf{K}}$&A non-unique matrix square root such that $\mathbf{K} = \sqrt{\mathbf{K}}\sqrt{\mathbf{K}}^T$\\
$u = u(\mathbf{x},t)$&the probability density function in state space\\
$x_i$&$i$th component of the state variable $\mathbf{x}$\\
$\mathbf{x}_0$&the initial condition for the mean of the estimation\\
$\mathbf{\Sigma}_0$&the initial condition for the covariance of the estimation\\
$\mathbf{x}_i$&$i$th cubature or off-center points in the state space ($i=1..2d$)\\
$\mathbf{\Sigma}$&the covariance matrix in the estimation\\
$\mathbf{M}$&a square root factorization of the covariance matrix $\mathbf{\Sigma}$\\
$\mathbf{J}$&Jacobian matrix of the drift velocity field $\mathbf{v}$\\
$\bar{\bx}$&the mean value in an estimation, during \emph{time-update}\\
$\hat{\bx}$&the mean value in an estimation, after \emph{measurement-update}\\
\hline
\end{tabularx}
\label{lst:symbols}
\end{table}
\section*{Appendix C: List of symbols and notations}
All symbols and notations used in more than one locations are listed in Table~\ref{lst:symbols}.


%





\bibliography{ref}
\bibliographystyle{IEEEtran}
\end{document}